\documentclass[12pt]{amsart}
\usepackage[margin = 1in]{geometry}
\usepackage{subfig, graphicx}
\begin{document}
\title[Log Heston Model for VIX]{Log Heston Model for Monthly Average VIX}
\author{Jihyun Park, Andrey Sarantsev}
\address{University of Michigan, Ann Arbor; Medical School}
\email{jihyunp@med.umich.edu}
\address{Department of Mathematics \& Statistics; University of Nevada, Reno}
\email{asarantsev@unr.edu}
\date{\today}
\newtheorem{theorem}{Theorem}
\newtheorem{lemma}{Lemma}
\theoremstyle{definition}
\newtheorem{asmp}{Assumption}
\newtheorem{defn}{Definition}
\begin{abstract}
We model time series of VIX (monthly average) and monthly stock index returns. We use log-Heston model: logarithm of VIX is modeled as an autoregression of order 1. Our main insight is that normalizing monthly stock index returns (dividing them by VIX) makes them much closer to independent identically distributed Gaussian. The resulting model is mean-reverting, and the innovations are non-Gaussian. The combined stochastic volatility model fits well, and captures Pareto-like tails of real-world stock market returns. This works for small and large stock indices, for both price and total returns. 
\end{abstract}

\maketitle

\thispagestyle{empty}

\section{Introduction}

\subsection{Stochastic volatility model} For the stock market, the very basic and classic model is geometric random walk (or, in continuous time, geometric Brownian motion). These processes have increments of logarithms, called {\it log returns} as independent identically distributed Gaussian random variables $X_1, X_2, \ldots$ This assumes that the standard deviation of these log returns, called {\it volatility}, is constant over time. Empirically, however, volatility does depend on time. Stretches of high volatility (usually corresponding to financial crises and economic troubles) alternate with periods of low volatility. More recently, models of {\it stochastic volatility} were proposed: 
\begin{equation}
\label{eq:returns-basic}
X_t = Z_tV_t.
\end{equation}
Here, $Z_1, Z_2, \ldots$ are independent identically distributed (IID) random variables (often normal, possibly with nonzero mean). The volatility $V_t$ is modeled by some mean-reverting stochastic process. A version of these is when $V_t$ is a function of past $V_s$ and $X_s$ for $s < t$: {\it generalized autoregressive conditional heteroscedastic} (GARCH) models. The word {\it heteroscedastic} means variance changing with time; de facto this is the same as {\it stochastic volatility} (SV). However, in practice, the term {\it stochastic volatility} is usually reserved for the models when $V_t$ is driven by its own innovation (noise) terms, distinct from  $Z_t$. The simplest mean-reverting model is {\it autoregression of order 1}, denoted by \textsc{AR}(1), and known as the {\it Heston model}, \cite{Heston}: 
\begin{equation}
\label{eq:Heston}
V_t = \alpha + \beta V_{t-1} + W_t,
\end{equation}
where $\beta \in (0, 1)$, and $W_t$ are IID mean zero innovations. However, in our research below, we find that this model fits poorly. Thus we propose an alternative. One natural model comes to mind: By construction, the volatility is always positive. We take its logarithm:
\begin{equation}
\label{eq:ln-VIX-AR}
 \ln V_t = \alpha +\beta \ln V_{t-1} + W_t.
\end{equation}
Here, $\beta \in (0, 1)$; this condition ensures mean-reversion. Also, $W_t$ are independent identically distributed mean zero random variables.  In our research, this model works better than the original Heston model. We call it the {\it log-Heston model}. Innovations $Z_t$ and $W_t$ might be dependent and correlated. We refer the reader to \cite{Skeptics}, where a proof is given that volatility models have strong predictive power. In continuous time, the combined model~\eqref{eq:returns-basic} and~\eqref{eq:ln-VIX-AR} can be written as follows (for $b > 0$):
\begin{align*}
\mathrm{d}S_t & = \mu V_tS_t\,\mathrm{d}t + \sigma V_tS_t\,\mathrm{d}W_t,\\
\mathrm{d}\ln V_t &= (a - b\ln V_t)\,\mathrm{d}t + \tilde{\sigma}\,\mathrm{d}B_t,
\end{align*}
where $B = (B_t,\, t \ge 0)$ and $W = (W_t,\, t \ge 0)$ are two standard Brownian motions, possibly correlated: $\mathbb E[B_tW_t] = \rho t$ for some $\rho \in (-1, 1)$. 

Almost always in the literature, we assume one can observe only log returns $X_t$ but not the volatility $V_t$. This makes estimation difficult for both GARCH and SV models. We refer the reader to the foundational articles \cite{GARCH-article, SV-article}, textbooks \cite{SV-book, GARCH-book}, and references therein. A related model is considered in \cite{log-GARCH1, log-GARCH2}, where GARCH is modified to make linear regression for log volatility instead of variance. This model is similar to our log-Heston model, but it has more restrictions on innovations. Also, here we observe VIX directly, whereas in estimating GARCH-type models we need to imply the volatility. 

\subsection{Time frequency} Mostly the literature on stochastic volatility focuses on daily data. However, in this article, we are interested in monthly returns, for two reasons. First, returns with larger time step exhibit more regularity. It is easier to write a model of them with IID Gaussian innovations. Second, monthly returns are more meaningful for long-term investors such as retirees, college funds, and endowments. Usually, withdrawals or contributions to these accounts are made monthly, quarterly, or even less frequency. 

\begin{figure}[t]
\centering
\subfloat[VIX]{\includegraphics[width=8cm]{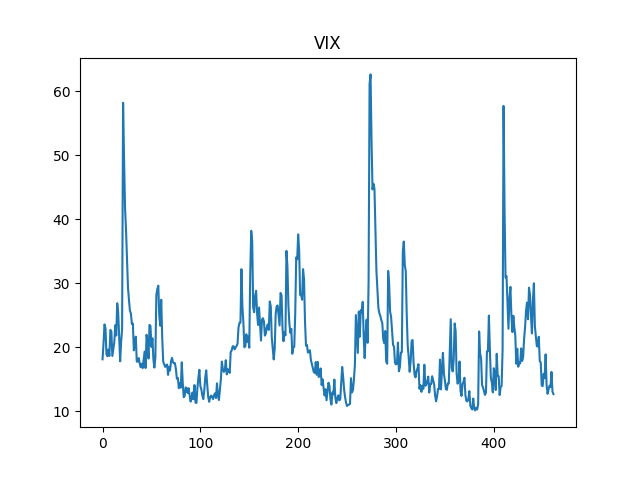}}
\subfloat[Small Price Returns]{\includegraphics[width=8cm]{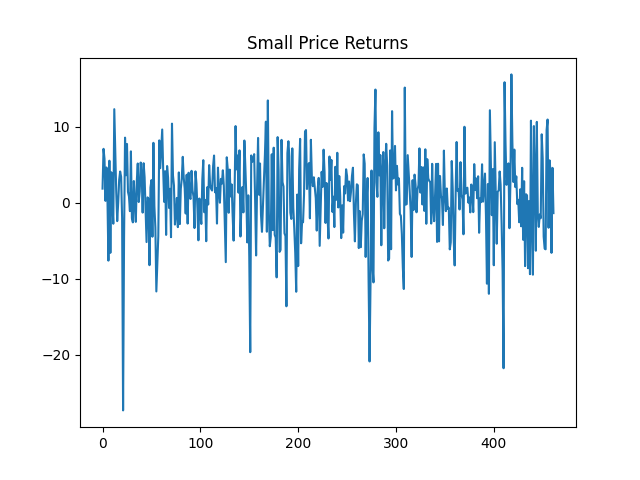}}
\caption{VIX Jan 1986-- Jun 2024} 
\label{fig:vix}
\end{figure}

\subsection{The volatility index: VIX} However, there is a rich options market upon common stocks and stock indices. A European option is a right to buy or sell a stock at a certain future time, called {\it maturity}, for a fixed price, called {\it strike}. The celebrated Black-Scholes formula for the price of an option uses the stock volatility as input, together with its current price, maturity, and strike. Now take the market price of an option with given maturity and strike upon a stock with certain (known) current price. Solve the Black-Scholes formula backwards to get volatility. This is called {\it implied volatility}. The well-known index Standard \& Poor 500 has particularly many such options traded on the market, with various maturities and strikes. Averaging implied volatilities computed from these options, we get the {\it Chicago Board of Options Exchange} (CBOE) {\it volatiltiy index} (VIX). Using this for $V_t$ makes estimating SV models much easier, because we have more data. We no longer have to imply stochastic volatility data from stock market returns. We can observe the former as well as the latter. See \cite{Whaley} about more details on the construction of VIX. This index was also studied in \cite{Hoerova}, where it was decomposed into two components, and their predictive power for stock market returns was studied. For VIX, we have daily data since 1986 (see below). But we work with monthly average data, since, as noted above, we are interested in monthly returns. 

\subsection{Organization of the article} In this article, we start by motivaing the use of VIX to normalize stock market returns in Section 2. We do this as follows: We compute several statistics for monthly stock index returns $Q_t$: skewness, kurtosis, and autocorrelation function. Next, we normalize these monthly returns $Q_t$ by dividing them by monthly average VIX $V_t$ and getting $Q_t/V_t$. Then we compute the same statistics for normalized stock index returns $Q_t/V_t$. It turns out that normalized returns are much closer to Gaussian independent identically distributed random variables than non-normalized returns. At the end of Section 2, we consider a slightly more complicated model with an intercept, see~\eqref{eq:stock-returns-advanced}. In Section 3, we fit the model~\eqref{eq:ln-VIX-AR} for monthly average VIX values $V_t$. This model fits well in the sense that residuals $W_t$ are independent identically distributed. But they are not Gaussian. We estimate the distribution of residuals and prove finite moment properties using the Hill estimator. Finally, in Section 4, we study long-term properties and finiteness of moments for this stochastic volatility log-Heston model. We reproduce the well-known property of real stock market returns: Pareto-like tails. 

The data is taken from public financial data libraries: Federal Reserve Economic Data and Kenneth French's Dartmouth College Financial Data. All code and data are on GitHub: \texttt{asarantsev/stocks-vix} repository.

\section{Statistical Motivation for Normalization}

\subsection{Data description} The main data series in this article is Chicago Board of Options Exchange (CBOE) Volatiltiy Index (VIX). This is taken from the Federal Reserve Economic Data (FRED) web site. The main series designed for Standard \& Poor 500 started in January 1990. This series is code-named VIXCLS on the FRED web site. Another series starts from January 1986, but was discontinued in November 2021. It is based on a related index, Standard \& Poor 100. This series is code named VXOCLS on the FRED web site. For each time series, we consider its monthly average values. When they overlap, these two indices are highly correlated. We unite them in one time series: Jan 1986 -- Feb 1990 for VXOCLS and Mar 1990 - Jan 2024 for VIXCLS. This is called $V_t$. We plot this data 1986--2024 for $T = 462$ months in \textsc{Figure}~\ref{fig:vix} (A). 

Next, we have data on monthly stock index returns. A {\it capitalization} of a stock is its market size (price of one stock times the number of outstanding stocks). We take two portfolios of stocks: Top 30\% and Mid 40\% (ranked by capitalization). Each portfolio is {\it capitalization-weighted}: Each stock in this portfolio is included in proportion to its size. We note that the returns of the Top 30\% portfolio closely correspond to that S\&P 500 or Russell 1000 (large US stocks), and Mid 40\% to Russell 2000 (small US stocks). This is why we call Top 30\% Large Stocks and Mid 40\% Small Stocks. The data is in percentages. As for the VIX, the data is Jan 1986 -- Jun 2024, total $T = 462$ months. We consider two versions of returns: {\it price returns}, which include only price movements, and {\it total returns}, which include dividends. Of course, total returns are always greater than or equal to price returns. Small price returns are on \textsc{Figure}~\ref{fig:vix} (B). 

\subsection{Description of statistical functions} Below we analyze whether $Q_1, \ldots, Q_T$ can reasonably be described as independent identically distributed (IID) Gaussian random variables, or {\it Gaussian white noise}. We test separately whether these are: (a) Gaussian; (b) IID. To this end, we compute the following statistical functions of returns $Q_1, \ldots, Q_T$. See the background on normality and IID testing in any standard time series textbook, for example \cite{Fan-Yao}: \cite[Section 1.4]{Fan-Yao} for testing ACF, and \cite[Section 1.5]{Fan-Yao} for normality testing. 

\begin{figure}[t]
\centering
\subfloat[$Q_t$]{\includegraphics[width=5cm]{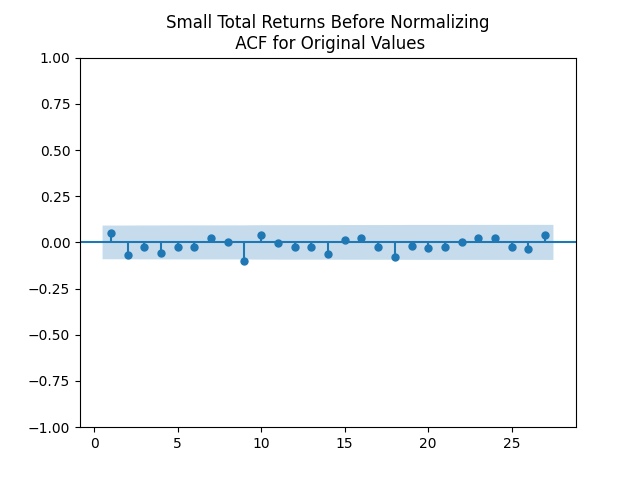}}
\subfloat[$|Q_t|$]{\includegraphics[width=5cm]{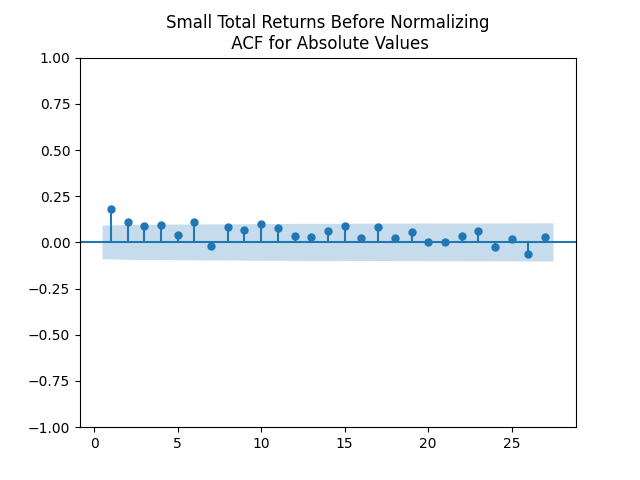}}
\subfloat[$Q_t$]{\includegraphics[width=5cm]{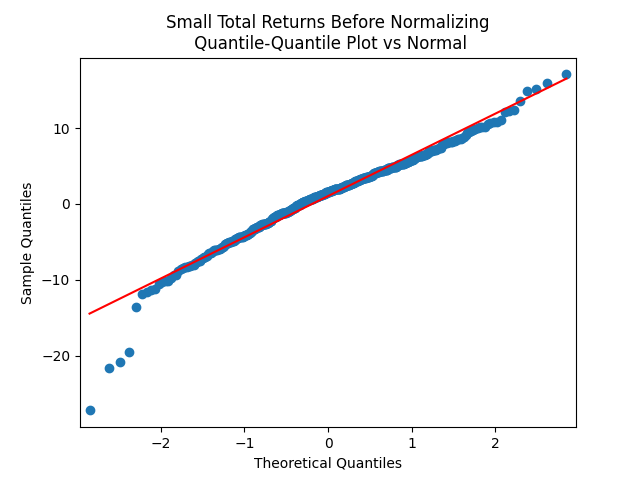}}
\newline
\subfloat[$Q_t/V_t$]{\includegraphics[width=5cm]{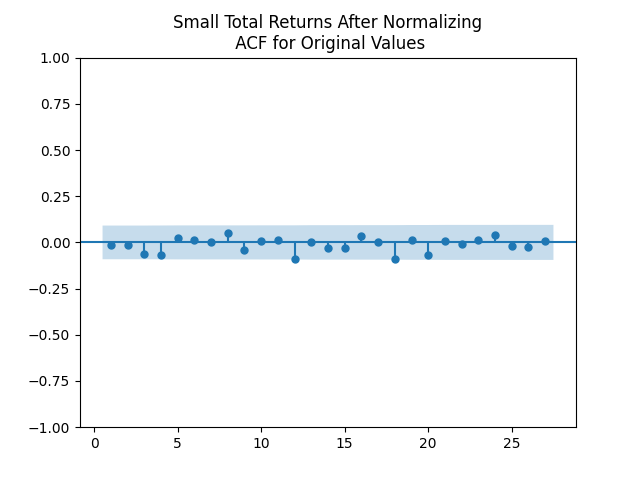}}
\subfloat[$|Q_t/V_t|$]{\includegraphics[width=5cm]{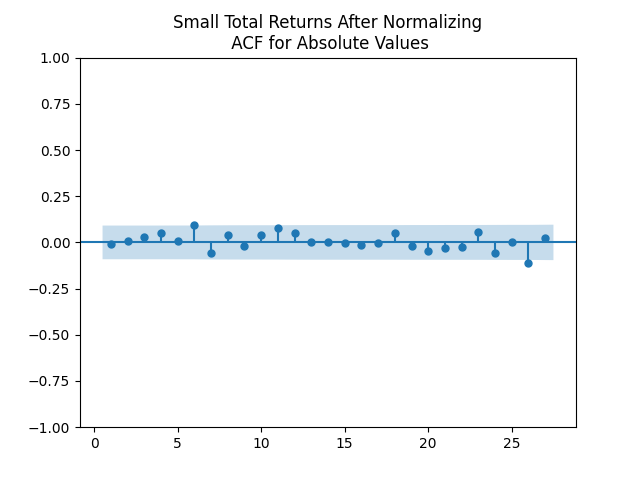}}
\subfloat[$Q_t/V_t$]{\includegraphics[width=5cm]{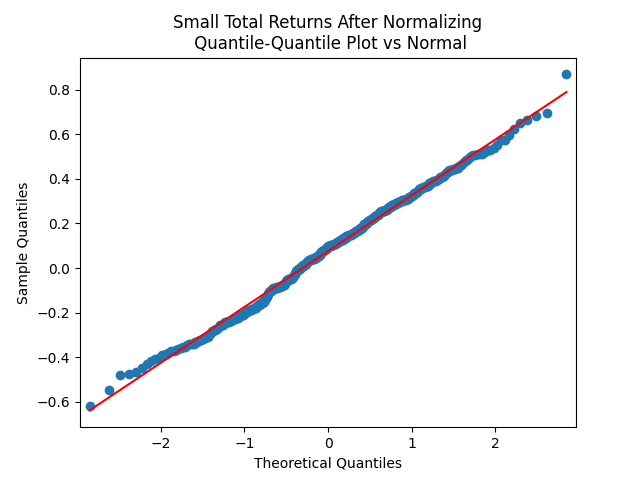}}
\caption{Statistical analysis of Small Total Returns: $Q_1, \ldots, Q_T$ before normalizing, and $Q_1/V_1, \ldots, Q_T/V_T$ after normalizing}
\label{fig:ACF-QQ}
\end{figure}

\begin{table}
\begin{tabular}{|c|c|c|c|c|}
\hline
Index Returns & $\mathfrak{s}$ & $\kappa$ & $|\!|\rho|\!|$  & $|\!|\tilde{\rho}|\!|$ \\
Small Total & -0.72/-0.02 & 2.52/-0.40 & 0.0121/0.0096 & 0.0629/0.0038\\
Large Total & -0.67/-0.11 &  1.73/-0.33 & 0.0072/0.0127 & 0.104/0.0167\\
Small Price & -0.72/-0.02 & 2.54/-0.4 & 0.0121/0.0101 & 0.0668/0.0043\\
Large Price & -0.67/-0.11 & 1.75/-0.33 & 0.0072/0.0129 & 0.1145/0.015\\
\hline
\end{tabular}
\caption{Statistical functions for stock index monthly returns $Q_t$ before/after normalizing, that is, dividing by monthly average VIX $V_t$. It is clear that dividing by VIX makes these functions closer to zero and therefore makes returns closer to Gaussian IID.}
\label{table:summary}
\end{table}

For (a), we compute {\it skewness} $\mathfrak{s}$ and {\it kurtosis} $\kappa$. Their well-known definitions are below. 
\begin{align}
\label{eq:skewness-kurtosis}
\mathfrak{s} := \frac{\hat{m}_3}{s^3}&\quad \mbox{and}\quad \kappa := \frac{\hat{m}_4}{s^4} - 3, \quad \mbox{where} \quad  \hat{m}_k := \frac1T\sum\limits_{t=1}^T(Q_t - \overline{Q})^k
\end{align}
is the empirical $k$th order centered moment, and $\overline{Q}$ and $s$ are the empirical mean and standard deviation. For i.i.d. $Q_1, \ldots, Q_T \sim Q$, these approximate true skewness and kurtosis of $Q$: 
\begin{equation}
\label{eq:true-skewness-kurtosis}
\mathfrak{s}(Q) := \frac{\mathbb E[(Q - \mathbb E[Q])^3]}{\mathrm{Var}^{3/2}(Q)},\quad \kappa(Q) := \frac{\mathbb E[(Q - \mathbb E[Q])^4]}{\mathrm{Var}^2(Q)} - 3.
\end{equation}
For a normal distribution $Q \sim \mathcal N(\mu, \sigma^2)$, we have: $\mathfrak{s}(Q) = \kappa(Q) = 0$. Thus if empirical skewness and kurtosis are closer to 0 this implies returns are closer to Gaussian. 

For (b), we compute the (empirical) {\it autocorrelation function} (ACF) $\rho(k)$ (for $k = 1, \ldots, 5$). This is the empirical correlation between $Q_1, \ldots, Q_{T-k}$ and $Q_{k+1}, \ldots, Q_T$. If $Q_1, \ldots, Q_T$ are IID, then for every $k = 1, \ldots, 5$, the theoretical ACF is $\mathrm{corr}(Q_j, Q_{j+k}) = 0$ (for any $j$). And the empirical ACF $\rho(k)$ is close to this theoretical value, which is $0$. To create a summary statistic, we then compute the $L^2$-norm (sum of squares) for this ACF. Incidentally, this is the Box-Pierce statistic, see \cite[Section 1.4.1]{Fan-Yao} or any other time series textbook. 
$$
|\!|\rho|\!| := \sum\limits_{k=1}^{5}\rho^2(k).
$$
Furthermore, we do the same for absolute values of returns: $|Q_1|, \ldots, |Q_T|$. The new ACF is denoted by $\tilde{\rho}(k)$, and we compute the $L^2$-norm for this new ACF: 
$$
|\!|\tilde{\rho}|\!| := \sum\limits_{k=1}^{5}\tilde{\rho}^2(k).
$$

\subsection{Results of statistical analysis} In \textsc{Table}~\ref{table:summary}, we have the results for original (non-normalized) returns $Q_1, \ldots, Q_T$ and for normalized returns $Q_1/V_1, \ldots, Q_T/V_T$ (after dividing by VIX). We see that the skewness and kurtosis for normalized retruns are much closer to zero. In the ACF for original values, there is not much improvement when switching from $Q_1, \ldots, Q_T$ to $Q_1/V_1, \ldots, Q_T/V_T$. But there is a lot of improvement, judging by the ACF for absolute values. The $L^1$-norm for these values become much closer to zero. 

\begin{table}
\begin{tabular}{|c|c|c|c|c|}
\hline
Returns & Small Total & Large Total & Small Price & Large Price\\
Mean & 0.075 & 0.072 & 0.069 & 0.062 \\
Stdev & 0.25 & 0.203 & 0.249 & 0.202\\
Correl & -54\% & -53\% & -54\% & -53\% \\
\hline
\end{tabular}
\caption{Analysis of normalized returns $Z_t = Q_t/V_t$ in~\eqref{eq:returns-basic}: mean, standard deviation, correlation with autoregression residuals $W_t$}
\label{table:mean-stdev}
\end{table}

\begin{table}
\begin{tabular}{|c|c|c|c|c|}
\hline
Returns & Small Total & Large Total & Small Price & Large Price\\
\hline
Coefficient $\theta$ & 3.6655 &3.3981 & 3.5628 & 3.2224 \\
Coefficient $\mu$ & -0.1304 & -0.1191 & -0.1316 & -0.1195 \\
Stdev $\sigma$ of $Z_t$ & 0.2421 & 0.1945 & 0.2416 & 0.1941 \\
\hline
Correlation with $W_t$ & -44\% & -42\% & -44\% & -42\% \\
Skewness $\mathfrak s$ for $Z_t$ & 0.022 & 0.024 & 0.026 & 0.026 \\
Kurtosis $\kappa$ for $Z_t$ & -0.392 & -0.275 & -0.391 & -0.27 \\
$|\!|\rho|\!|$ & 0.009 & 0.0177 & 0.0092 & 0.0177 \\
$|\!|\tilde{\rho}|\!|$ & 0.0174  & 0.0134 & 0.0177 & 0.0149 \\
\hline
\end{tabular}
\caption{Regression~\eqref{eq:stock-returns-regression} results and snalysis of residuals $Z_t$. Comparing with \textsc{Table}~\ref{table:summary}, we see that skewness, kurtosis, and ACF (measured by $|\!|\rho|\!|$ and $|\!|\tilde{\rho}|\!|$) are as close or closer to zero as normalized returns. Thus it makes sense to model regression residuals as Gaussian IID.}
\label{table:regression}
\end{table}

This is supported by the graphs of empirical ACF for returns before and after normalizing, as well as by the quantile-quantile plots vs the normal distribution. In \textsc{Figure}~\ref{fig:ACF-QQ}, we see ACF $\rho(k)$ for several $k$; ACF $\tilde{\rho}(k)$ for several $k$; and the quantile-quantile plot vs the normal distribution, for the original Small Total Returns $Q_1, \ldots, Q_T$. \textsc{Figure}~\ref{fig:ACF-QQ} contains the same for the normalized version: $Q_1/V_1, \ldots, Q_T/V_T$. Plots for Small Price Returns, Large Total Returns, and Large Price Returns look similarly. One sees that the ACF plots for $Q_t$ in \textsc{Figure}~\ref{fig:ACF-QQ}~\textsc{(A)} and for $Q_t/V_t$ in \textsc{Figure}~\ref{fig:ACF-QQ}~\textsc{(D)} look both close to zero, but there is a noticeable difference in the ACF plots for $|Q_t|$  in \textsc{Figure}~\ref{fig:ACF-QQ}~\textsc{(B)} and for $|Q_t/V_t|$ in \textsc{Figure}~\ref{fig:ACF-QQ}~\textsc{(E)}. The former has some values away from zero, but the latter looks closer to zero. Also, the quantile-quantile plots in \textsc{Figure}~\ref{fig:ACF-QQ}~\textsc{(C)} and~\textsc{(F)} show that normalized returns are closer to the normal than non-normalized ones. To the authors, it is a positive surprise that the volatility created from S\&P 500 (Large Price Returns) works well also for Small Stocks (Price or Total Returns), and Large Total Returns. 

We also refer the reader to a more extensive study \cite{Memory} of autocorrelation for S\&P 500 price returns. They show that raising absolute returns to a power produces more autocorrelation. However, their research is for daily, not monthly returns. 

Below, we provide \textsc{Table}~\ref{table:mean-stdev} of means $\mu$ and standard deviations $\sigma$ for four cases: Small Total Returns, Small Price Returns, Large Total Returns, Large Price Returns. To summarize, the model~\eqref{eq:returns-basic} with $Z_t$ IID Normal seems to be fitting well.

\subsection{Complete regression} A more general version of the model~\eqref{eq:returns-basic} for $Q_t$ has an intercept:
\begin{equation}
\label{eq:stock-returns-advanced}
Q_t = \theta + V_t(\mu + Z_t),\quad Z_t \sim \mathcal N(0, \sigma^2).
\end{equation}
We can rewrite~\eqref{eq:stock-returns-advanced} as follows:
\begin{equation}
\label{eq:stock-returns-regression}
\frac{Q_t}{V_t} = \frac{\theta}{V_t} + \mu + Z_t,\quad Z_t \sim \mathcal N(0, \sigma^2).
\end{equation}
Then we fit this regression for each of the four stock return series: Small Total, Small Price, Large Total, Large Price. The summary is in \textsc{Table}~\ref{table:regression}. Next, for residuals $Z_t$, their skewness and kurtosis are as small or smaller than the second ones in \textsc{Table}~\ref{table:summary}. This implies $Z_t$ are as well or even better modeled by the normal distribution than normalized returns. The coefficients $\theta$ and $\mu$ are significantly different from zero, with Student $t$-test $p$-values less than $0.002$. The model~\eqref{eq:stock-returns-advanced} will be our main model in the rest of the article.

\section{Time Series Analysis of Monthly Average VIX}

\subsection{The simplest Heston model} As discussed in the Introduction, the simplest model for VIX is {\it Heston model}~\eqref{eq:Heston}. Fitting this for monthly Jan 1986 -- Jun 2024 data gives us $\alpha = 3.097$, $\beta = 0.844$. However, there is a problem. \textsc{Figure}~\ref{fig:ACF-Heston} (A, B) has plots of the autocorrelation function (ACF) for innovations $W$ and for $|W|$. These show that $W$ are not quite white noise. In particular, ACF with lags 1 and 2 for $|W|$ are outside of the shaded region, corresponding to white noise hypothesis. Thus we switch to find another model.

\begin{figure}[t]
\centering
\subfloat[ACF for $W$ in~\eqref{eq:Heston}]{\includegraphics[width=7cm]{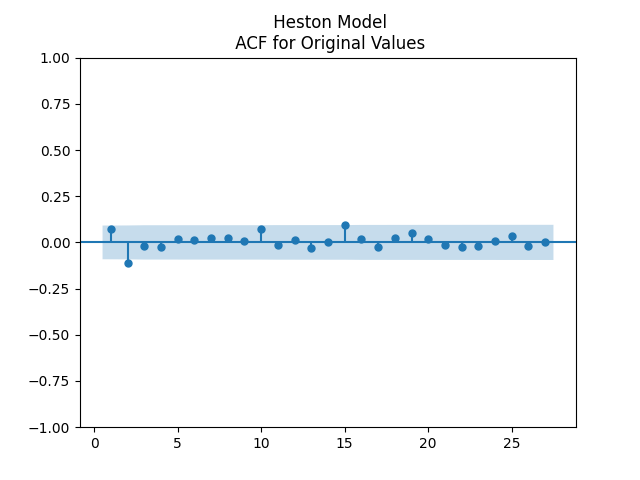}}
\subfloat[ACF for $|W|$ in~\eqref{eq:Heston}]{\includegraphics[width=7cm]{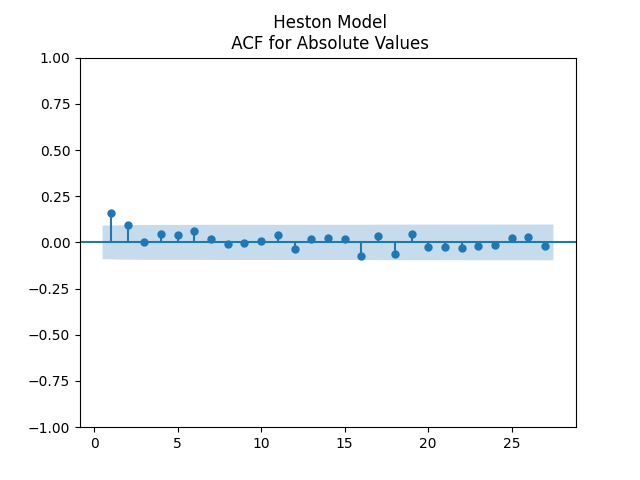}}
\newline
\subfloat[ACF for $W$ in~\eqref{eq:dln-VIX}]{\includegraphics[width=7cm]{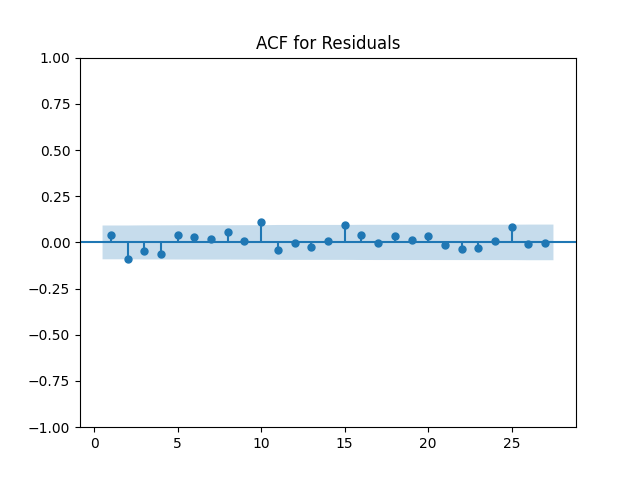}}
\subfloat[ACF for $|W|$ in~\eqref{eq:dln-VIX}]{\includegraphics[width=7cm]{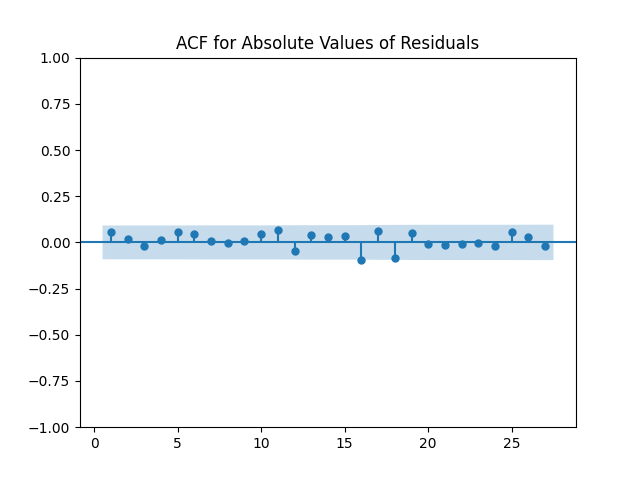}}
\caption{Analysis of Residuals $W$ in~\eqref{eq:Heston} and in~\eqref{eq:dln-VIX}}
\label{fig:ACF-Heston}
\end{figure}

\subsection{Autoregression of order 1 in log scale} As discussed in the Introduction, the VIX is always positive and mean-reverting. Thus it makes sense to use autoregression of order 1 on log scale, as in~\eqref{eq:ln-VIX-AR}. Rewrite this as
\begin{equation}
\label{eq:dln-VIX}
\triangle \ln V_t = \alpha + (\beta-1)\ln V_{t-1} + W_t.
\end{equation}
Fitting this for monthly Jan 1986 -- Jun 2024 data gives us $\alpha = 0.346$, $\beta - 1 = -0.118$ (therefore $\beta = 0.882$).  The correlation between $\triangle\ln V_t$ and $\ln V_{t-1}$ is $-24\%$, therefore $R^2 = 5.8\%$. The $p$-value for Student $t$-test with zero correlation as null hypothesis is $1.5\cdot 10^{-7}$. Note that zero correlation in~\eqref{eq:dln-VIX} corresponds to $\beta = 1$, so $\ln V_t$ becomes a random walk. The Augmented Dickey-Fuller unit root test applied to $\ln V$ with 15 lags gives us $p = 0.26\%$. Thus we reject the unit root null hypothesis. In sum, the evidence is strongly in favor of mean-reversion in $\ln V_t$. 

\textsc{Figure}~\ref{fig:ACF-Heston} (C, D) has plots of the autocorrelation function (ACF) for innovations $W$ and for $|W|$. These show clearly that $W$ can be reasonably modeled as independent identically distributed mean zero random variables. As discussed earlier with regard to stock market returns, we need an additional ACF plot for $|W|$ and are not satisfied with only the ACF plot for $W$. Indeed, it is possible for time series to have zero linear autocorrelation but depend upon the past in a nonlinear way. Stochastic Volatility models $X_t = V_tZ_t$ described in the Introduction are a perfect example of that, if $Z_t$ is independent of the innovations in $V_t$. 

However, the distribution of the innovations $W$ is not Gaussian, but has heavier tails. \textsc{Figure}~\ref{fig:gof-vixres} (A) clearly shows this using the quantile-qauntile plot vs the normal distribution. Computed using~\eqref{eq:skewness-kurtosis}, the skewness of $W$ is 2, and the excess kurtosis is 9. Of course, we reject the null normality hypothesis based on Shapiro-Wilk and Jarque-Bera tests. 

\subsection{Variance-gamma distribution of regression residuals} Instead, we fit variance-gamma distribution using Python package 
\texttt{dlaptev$\backslash$VarGamma} from \texttt{GitHub}. We use the parametrization from \cite{Seneta}, since that Python package uses it as well. 

\begin{defn}
The {\it variance-gamma} (VG), or {\it generalized asymmetric Laplace} (GAL), distribution is defined as the distribution of a random variable 
\begin{equation}
\label{eq:vg-defn}
X = c + a\Gamma + b\sqrt{\Gamma}Y,
\end{equation}
where $Y \sim \mathcal N(0, 1)$ and $\Gamma$ is independent of $Y$ and has a gamma distribution with density
$$
\frac{(1/\nu)^{1/\nu}}{\Gamma(1/\nu)}x^{1/\nu - 1}e^{-x/\nu},\quad x > 0.
$$
\end{defn}

Here, $\Gamma$ has shape parameter $\nu^{-1}$, and is rescaled so that $\mathbb E[\Gamma] = 1$. The moment generating function of this distribution is given by
\begin{equation}
\label{eq:finite-MGF}
M(t) = \mathbb E[e^{tX}] = e^{ic\omega}\left(1 - a\nu t - b^2\nu t^2/2\right)^{-1/\nu}.
\end{equation}
and is well-defined if $1 - \nu at - b^2\nu t^2/2 > 0$. More properties on this distribution, along with different parametrizations, are available in \cite{VG-review}. A multivariate version is discussed in \cite{GAL}. The package \texttt{dlaptev$\backslash$VarGamma} provides the maximum likelihood estimation of parameters, with initial step by method of moments done in \cite{Seneta}. For the residuals $W$ from~\eqref{eq:ln-VIX-AR}, the parameter estimation gives us 
\begin{equation}
\label{eq:est-params}
a = 0.0621, \quad b = 0.1392, \quad c = -0.0621, \quad \nu = 0.6573.
\end{equation}
Thus the MGF from~\eqref{eq:finite-MGF} is well-defined for $-16.1 < t < 9.7$. The QQ plot of the residuals $W$ vs simulated variance-gamma random variable~\eqref{eq:vg-defn} with parameters~\eqref{eq:est-params}, and the PP plot of $W$ vs this variance-gamma distribution, are given in \textsc{Figure}~\ref{fig:gof-vixres} (B) and (C). They show somewhat good but nor perfect fit though.

\begin{figure}[t]
\centering
\subfloat[QQ vs N]{\includegraphics[width = 5cm]{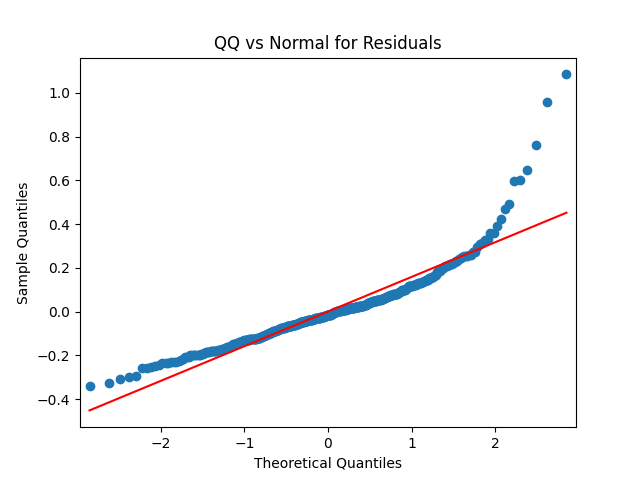}}
\subfloat[PP vs VG]{\includegraphics[width=5cm]{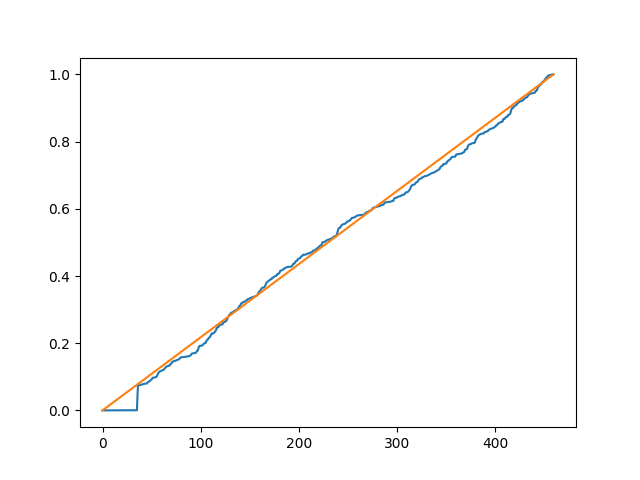}}
\subfloat[QQ vs VG]{\includegraphics[width=5cm]{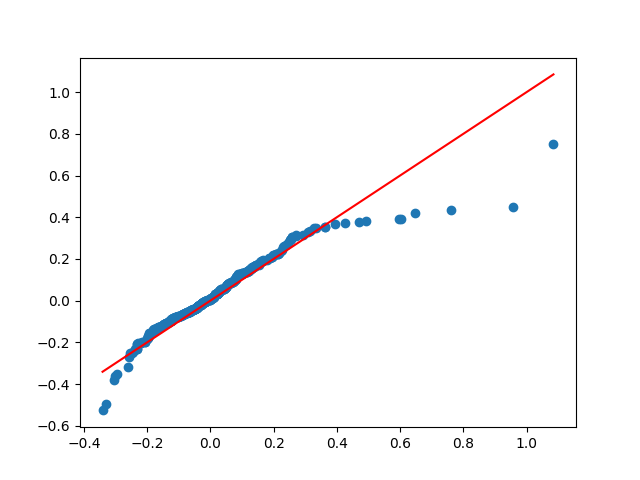}}
\bigskip
\caption{Plots for autoregression innovations $W$ in~\eqref{eq:ln-VIX-AR} vs Normal and Variance-Gamma (VG). Poor fit for Normal, but better for Variance-Gamma.} 
\label{fig:gof-vixres}
\end{figure}

However, we see that the VG model does not fit perfectly. There are problems in the right tail. At least we can estimate left and right tails using Hill estimator from \cite{Hill} for $e^{W_1}, \ldots, e^{W_T}$. We sort these: $e^{W_{(1)}} < \ldots < e^{W_{(T)}}$, and choose the cutoff $r$, and compute the left and right tail estimates $\gamma^-_r$ and $\gamma^{+}_r$, respectively. 
\begin{align}
\label{eq:Hill}
\begin{split}
\left(\gamma^-_r\right)^{-1} &:= \frac1r\sum\limits_{t=1}^r\ln e^{W_{(t)}} - \ln e^{W_{(r+1)}} = \frac1r\sum\limits_{t=1}^r W_{(t)} - W_{(r+1)},\\
\left(\gamma^+_r\right)^{-1} &:= -\frac1r\sum\limits_{t=T-r+1}^{t=T}\ln e^{W_{(t)}} + \ln e^{W_{(T-r)}} = \frac1r\sum\limits_{t=T-r+1}^T W_{(t)} - W_{(T-r)}.
\end{split}
\end{align}
An interesting question is how to choose the cutoff $r$. We plot $r$ vs $\gamma^-_r$ and $r$ vs $\gamma^+_r$, see \textsc{Figure}~\ref{fig:Hill}. We choose $r = 100$ and get $\gamma^+_r = 7.3$ and $\gamma^-_r = 15.7$. Therefore, 
$$
\mathbb E[e^{tW}] < \infty, \quad -14.7 < t < 6.3.
$$

\begin{figure}[t]
\centering
\subfloat[Left Tail $\gamma^-_r$ vs $r$]{\includegraphics[width = 7cm]{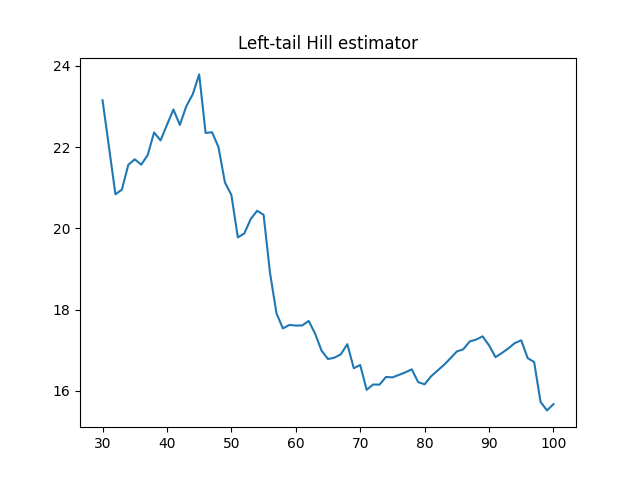}}
\subfloat[Right Tail $\gamma^+_r$ vs $r$]{\includegraphics[width = 7cm]{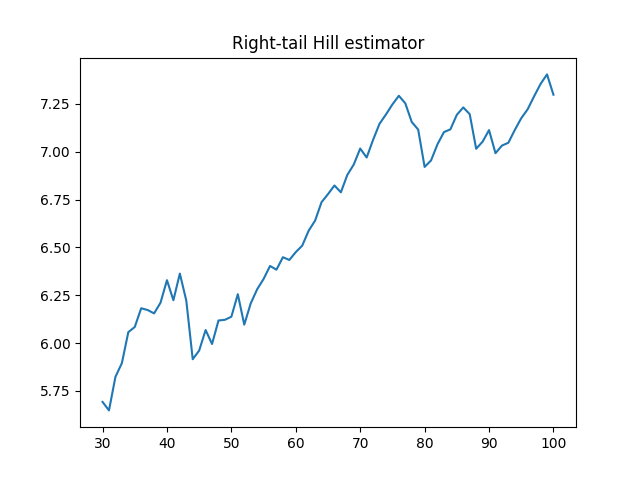}}
\caption{Hill Estimator~\eqref{eq:Hill} for $e^{W}$ vs $r$}
\label{fig:Hill}
\end{figure}

To conclude: Statistical analysis shows that $\ln V_t$ is well-modeled by mean-reverting autoregression of order $1$. The residuals of this autoregression are independent identically distributed but not Gaussian. However, they have exponential-like tails, with $\mathbb E[e^{tW}] < \infty$ for $|t| < t_0$ for some $t_0 > 0$. 

\section{Combined Model of Volatility and Stock Returns}

Put together equations for VIX and stock returns from~\eqref{eq:stock-returns-advanced} and~\eqref{eq:ln-VIX-AR}:
\begin{align}
\label{eq:double}
\begin{split}
\ln V_t &= \alpha + \beta \ln V_{t-1} + W_t,\\
Q_t &= \theta + V_t(\mu + Z_t).
\end{split}
\end{align}

As discussed in Section 2, $Z_t$ can be modeled as IID normal random variables: $Z_t \sim \mathcal N(\mu, \sigma^2)$. As discussed in Section 3, however, $W_t$ are IID mean-zero but not normal.

\subsection{Finite moments for volatility} First, let us deal just with the volatility modeled by the log Heston model~\eqref{eq:ln-VIX-AR}.

\begin{asmp}
In~\eqref{eq:ln-VIX-AR}, $0 < \beta < 1$, and $W_1, W_2, \ldots$ are independent identically distributed with $\mathbb E[W_t] = 0$ and $M_W(u) := \mathbb E[e^{uW_t}] < \infty$ for a certain $u > 0$; the initial condition $V_0$ is independent of all $W_t$, with $\mathbb E[V_0^u] < \infty$. 
\label{asmp:W}
\end{asmp}

\begin{lemma}
Under Assumption~\ref{asmp:W}, we have: $\sup_{t \ge 0}\mathbb E[V^u_t] < \infty$. Moreover, there exists a unique stationary distribution $\pi_{\infty}^V$, and the variable $V_{\infty} \sim \pi_{\infty}^V$ satisfies $\mathbb E[V^{u}_{\infty}] < \infty$. 
\label{lemma:finite-moment}
\end{lemma}

\begin{proof}
Without loss of generality, assume $a = 0$. We rewrite~\eqref{eq:ln-VIX-AR} as
$$
\ln V_t = \beta^t\ln V_0 + \sum\limits_{s=1}^{t}\beta^{t-s}W_s. 
$$
Write this in exponential form:
$$
V^u_t = V_0^{\beta^tu}\cdot\prod\limits_{s=1}^t\exp\left[u\beta^{t-s} W_s\right].
$$
Apply the expectation, and use independence of $W_1, \ldots, W_t$:
\begin{equation}
\label{eq:prod-exp}
\mathbb E[V^u_t] = \mathbb E[V_0^{b^tu}] \cdot\prod\limits_{s=1}^t\mathbb E\left[\exp(u\beta^{t-s}W_s)\right].
\end{equation}
By Assumption~\ref{asmp:W}, $M_W(w)$ is well-defined for $0 \le w \le u$, in particular for $w = \beta^{t-s}u$ since $\beta \in (0, 1)$. Also, $M_W(0) = 1$, $M'_W(0) = \mathbb E[W] = 0$, $M''_W(0) = \mathbb E[W^2] := \sigma^2_W < \infty$. By Taylor's formula, we get: As $w \to 0$, 
\begin{equation}
\label{eq:Taylor}
\ln M_W(w) \sim M_W(w) - 1 \sim 0.5\sigma^2_Ww^2.
\end{equation}
Next, we take logarithms in~\eqref{eq:prod-exp}:
\begin{equation}
\label{eq:logs}
\ln \mathbb E[V^u_t] = \ln\mathbb E[V_0^{\beta^tu}] + \sum\limits_{s=1}^t\ln  M_W(u\beta^{t-s}) =  \ln\mathbb E[V_0^{\beta^tu}] + \sum\limits_{s=0}^{t-1}\ln  M_W(u\beta^{s}).
\end{equation}
As $t \to \infty$, we have: $u\beta^t \to 0$, and $\ln \mathbb E[V^u_t] \to \ln 1 = 0$. Moreover, $u\beta^s \to 0$ for $s \to \infty$. Applying~\eqref{eq:Taylor}, we get: $M_W(u\beta^s) \sim 0.5\sigma^2_Wu^2\beta^{2s}$ as $s \to \infty$. Therefore, both series below converge: 
$$
\sum 0.5\sigma^2_Wu^2\beta^{2s}\quad \mbox{and}\quad \sum\ln M_W(u\beta^s).
$$
Thus, the sequence~\eqref{eq:logs} converges as $t \to \infty$. Every converging sequence is bounded. This proves the first claim of Lemma~\ref{lemma:finite-moment}. Let us show the second claim. This stationary distribution is, in fact, the weak limit: $V_t \to V_{\infty}$. By the Skorohod representation theorem, we can switch to almost sure convergence on a new probability space. But $\sup_{t \ge 0}\mathbb E[V^u_t] < \infty$. By Fatou's lemma, we get: $\mathbb E[V^u_{\infty}] < \infty$. This proves the second claim of Lemma~\ref{lemma:finite-moment}.
\end{proof}

\subsection{Finite moments for stock returns} Now, we turn back to stock index returns $Q$. Under Assumption~\ref{asmp:W}, we prove they have finite moments up to a certain order. This captures a well-known property of stock market returns: Pareto-like tails, first noted by Fama in \cite{Fama}. However, the second moment of real-world stock returns is finite. This follows from Theorem~\ref{thm:moment} for $u = 2 + \varepsilon$ and the empirical analysis in the previous section.

\begin{asmp}  $(Z_t, W_t)$ are IID random vectors with mean $(0, 0)$, while the marginal distribution of $Z$ is Gaussian. 
\label{asmp:S}
\end{asmp}

We do {\it not} require $Z_t$ and $W_t$ to be independent or even uncorrelated, though this would simplify computations; see later.

\begin{theorem}
Under Assumptions~\ref{asmp:W} and~\ref{asmp:S}, the system~\eqref{eq:double} has a unique stationary solution $(V_{\infty}, Q_{\infty})$ with $\mathbb E[|Q_t|^w] < \infty$ and $\mathbb E[|Q_{\infty}|^w] < \infty$ for any $w \in (0, u)$.
\label{thm:moment}
\end{theorem}

\begin{proof} That the system has a stationary solution is clear. Now, by Minkowski's inequality, 
\begin{equation}
\label{eq:ineq-1}
\left(\mathbb E[|Q_t|^w]\right)^{1/w} \le \theta + \mu\left(\mathbb E[V_t^w]\right)^{1/w} + \left(\mathbb E[V_t^wZ_t^w]\right)^{1/w}.
\end{equation}
We already have $\mathbb E[V_t^w] < \infty$. It suffices to prove $\mathbb E[V_t^wZ_t^w] < \infty$. Pick $p = u/w > 1$ and let $q > 1$ be such that $1/p + 1/q = 1$. Then by H\''older's inequality, 
\begin{equation}
\label{eq:ineq-2}
\mathbb E[|Q_t|^w] \le \left(\mathbb E\left[|V_t|^{wp}\right]\right)^{1/p}\cdot \left(\mathbb E\left[|Z_t|^{wq}\right]\right)^{1/q} = 
\left(\mathbb E\left[|V_t|^u\right]\right)^{1/p}\cdot  C(wq)^{1/q},
\end{equation}
where $C(a)$ is the $a$th absolute moment of the random variable $Z$: $C(a) := \mathbb E[|Z|^a]$. Therefore, this $w$th moment of $Q_t$ is finite. We can take $\sup_{t \ge 0}$ in~\eqref{eq:ineq-1} and~\eqref{eq:ineq-2}, and repeat the proof. The same proof works for $Q_{\infty}$ instead of $Q_t$. 
\end{proof}

\subsection{Limit theorems} Below, under additional assumptions, we state and prove the Law of Large Numbers and the Central Limit Theorem for 
monthly stock index returns. Since these assumptions are borne out by real-world data, we see that the real-world property is reproduced: Over longer time periods, returns are closer to normal. See discussion of this in \cite{Aggregate} and references therein. 

\begin{asmp}
The distribution of $W_t$ has strictly positive density on the entire real line. Moreover, $\mathbb E\left[e^{2W_t}\right] < \infty$. 
\label{asmp:W-density}
\end{asmp} 

\begin{theorem} Under Assumptions~\ref{asmp:W},~\ref{asmp:S},~\ref{asmp:W-density}, the stock index returns satisfy the Law of Large Numbers: as $T \to \infty$, almost surely
\begin{equation}
\label{eq:LLN}
\overline{Q}_T := \frac1T(Q_1 + \ldots + Q_T) \to \mathbb E[Q_{\infty}],
\end{equation}
and the Central Limit Theorem:
\begin{equation}
\label{eq:CLT}
\sqrt{T}(\overline{Q}_T - \mathbb E[Q_{\infty}]) \to \mathcal N(0, \sigma^2_Q).
\end{equation}
\end{theorem}

\begin{proof} Under Assumptions~\ref{asmp:W},~\ref{asmp:S},~\ref{asmp:W-density}, the process $((V_t, Q_t), t \ge 0)$ is Markov on the state space $\mathcal X := (0, \infty)\times \mathbb R$ with $n$-step {\it transition function} $P^n((v, q), A) := \mathbb P((V_n, Q_n) \in A\mid (V_0, Q_0) = (v, q))$. The Markovian property follows from the fact that $(V_t, Q_t)$ is a one-to-one function of $(\ln V_t, Z_t)$. Indeed, $((\ln V_t, Z_t), t \ge 0)$ is a Markov process on $\mathbb R^2$. This is because $\ln V_t$ is an autoregression of order $1$, and $Z_t$ are IID. Thus it is a linear model. What is more, this process $((\ln V_t, Z_t), t \ge 0)$ has a unique stationary distribution, and is $V$-uniformly ergodic, with $V(x_1, x_2) := x_1^2 + x_2^2$. See the definitions and results in \cite[Chapter 16]{MeynTweedieBook}. The Law of Large Numbers and the Central Limit Theorem  follow from \cite[Chapter 17]{MeynTweedieBook}.
\end{proof}

\section{Conclusion} We propose and fit a log-Heston model for monthly average volatility and stock idnex returns. It fits actual financial data well. This model exhibits good long-term properties: stationarity and mean-reversion. It captures a well-known property of real-world stock idnex returns: Pareto-like tails. For this, we need innovations in this log-Heston model which have tails heavier than Gaussian. Future research might include finding bivariate distributions which are a good fit for innovations $(Z_t, W_t)$. Also, it is important to analyze other stock indices, such as Value, Growth, or international.

\end{document}